\documentclass{ifacconf}
\usepackage{amsfonts, amsmath, amssymb, bbm}
\usepackage{graphicx}      
\usepackage{natbib}        
\usepackage{xcolor}

\newtheorem{definition}{Definition}

\newtheorem{lemma}{Lemma}

\newtheorem{proposition}{Proposition}

\newtheorem{remark}{Remark}

\newcommand{\control}{\xi}
\newcommand{\statistics}{p}
\newcommand{\cost}{\gamma}

\newcommand{\tcb}{\textcolor{black}}

\newcommand{\mc}{\mathcal}
\newcommand{\de}{\mathrm{d}}
\begin{document}
\begin{frontmatter}

\title{Optimal Interventions on the Linear Threshold Model in Large-Scale Networks} 

\author[First]{Leonardo Cianfanelli}
\author[First]{Sebastiano Messina}  
\author[First]{Giacomo Como}
\author[First]{Fabio Fagnani}

\address[First]{Department of Mathematical Sciences, Politecnico di Torino, Corso Duca degli Abruzzi 24, 10129, Torino, Italy\\
$($Email: \{{leonardo.cianfanelli, sebastiano.messina, giacomo.como, fabio.fagnani}\}@polito.it.$)$}

\begin{abstract}                
    We study an optimal intervention problem on the linear threshold model (LTM) in which a social planner aims to design minimal-cost interventions that modify the agents' thresholds, under the constraint that at least a predefined fraction of agents reaches a given state after a finite number of iterations.  While this problem is known to be NP-hard and its exact solution requires full knowledge of the network structure, we focus on approximate solutions for large-scale networks and assume that the planner has only  statistical knowledge of the network. In particular, we build on a local mean-field approximation of the LTM that is known to hold true on large-scale random networks, and reformulate the optimal intervention problem as a linear program with an infinite set of constraints. We then show how to approximate the solutions of the latter problem by standard linear programs with finitely many constraints. \tcb{Finally, our approach is validated through numerical experiments on real-world networks and compared both with optimal seeding and state-of-the-art algorithms for the least-cost influence.}
\end{abstract}


\end{frontmatter}

\section{Introduction}
\tcb{
The study of the diffusion of behaviors, innovations, and opinions over social networks has become a central topic across sociology, economics, computer science, and control. A foundational contribution is the linear threshold model (LTM) introduced by the seminal work of \cite{Granovetter:78}, whereby individuals adopt a behavior once the fraction of adopting neighbors exceeds a personal threshold. This framework demonstrates how heterogeneous micro-level preferences can generate aggregate discontinuities and large-scale cascades.}

\tcb{The algorithmic study of diffusion in networks was significantly shaped by  \cite{Kempe.ea:03}, who first introduced the influence maximization problem (IMP), whereby a system planner aims to select a fixed number of seed nodes, with the objective of maximizing the asymptotic spread of adoption under the LTM over a social network. While the IMP is NP-hard, \cite{Kempe.ea:03} considered a randomized version, whereby the individual thresholds are sampled independently and uniformly from the interval $[0,1]$ and the system planner aims at maximizing the expected spread of the LTM over a given network. They proved that such randomized version of the IMP is a submodular problem, so that a low-complexity, iterative greedy algorithm exists achieving a $(1-1/e)$ approximation. Subsequent studies focused on how to optimize the seed choice in each round of the greedy algorithm and proposed algorithmic solutions (\cite{goyal2011simpath}).}

\tcb{In parallel, researchers investigated cost-minimization variants of the IMP. The Target Set Selection (TSS) problem, formalized by \cite{chen2009approximability}, is the dual problem, which seeks the smallest initial set of seeds that eventually activate the entire network under the LTM. We refer to  \cite{chen2013information} and \cite{Como.ea:22} for an overview on these problems and algorithmic solutions. Weighted extensions (WTSS) proposed in \cite{raghavan2019branch} and \cite{Cordasco:2015}  allow for heterogeneous costs.}

\tcb{A significant step toward economically meaningful interventions is the Least-Cost Influence Problem (LCIP) introduced in \cite{gunnec2012integrating} and subsequently analyzed in \cite{gunnecc2020branch,gunnecc2020least}, where the planner can assign fractional incentives to individuals. Rather than restricting interventions to binary seeding decisions, this framework allows partial reductions of individual thresholds, thereby modeling more realistic policies such as targeted discounts or subsidies, with the goal of minimizing the intervention cost while ensuring a maximal spread of the adoption. Related fractional formulations have also been explored by \cite{demaine2014influence} who generalize the IMP by seeking the maximum spread of adoption in a network with randomized thresholds with a given budget, while allowing fractional allocations that decrease the thresholds, highlighting that fractional allocations can substantially alter diffusion outcomes compared to the IMP, despite not increasing the computational complexity of the problem. Similar problems can be found in the literature with different names, see e.g., the Targeting with Partial Incentives (TPI) problem of \cite{Cordasco:2015}, who propose algorithms for its solution on networks with special topologies, such as the complete graph or trees, and heuristic solutions on arbitrary graphs.}

\tcb{The scalability of solutions to these problems remains a fundamental challenge in social networks, where the number of agents can be extremely large and only aggregate information may be available. At the same time, the study of diffusion processes on random networks has developed analytical tools based on configuration models (CMs) and branching process approximations. Mean-field techniques characterize cascade behavior in terms of degree distributions and threshold statistics, enabling tractable analysis of large populations without explicitly enumerating of nodes (see \cite{Rossi.ea:17} and \cite{Ravazzi}).}

\tcb{In this paper, we study the LCIP on general directed multigraphs and develop a formulation tailored to large-scale random networks. By leveraging local mean-field approximations for 
diffusion on the CM, we describe adoption dynamics at the statistical level rather than individual level. 
This allows us to recast the LCIP  as a linear program whose dimension is independent of the network size, subject to an infinite number of constraints.}
We then propose a discretization approach transforming it into a finite linear program, and prove that its solutions are feasible
and close to optimal for the original problem. In doing so, we find close to optimal and tractable solutions that do not scale with the network size.

\tcb{This approach was already explored by \cite{messina2024optimal} for the WTSS problem and \cite{cianfanelli2025optimal} for the design of interventions in opinion dynamics. Besides reducing complexity, it has further advantages: it does not require detailed structural information about the network, the threshold, or the intervention cost associated with each individual agent, which are often unavailable to the planner, because of their cost or privacy restrictions.}
\tcb{It is worth stressing that our use of randomness differs substantially from that of \cite{Kempe.ea:03} and \cite{demaine2014influence}, both because we consider random networks rather than deterministic ones, and since we can account for arbitrary joint distributions of thresholds and degrees rather than being restricted to independent uniform random thresholds. The solutions obtained by our approach have probabilistic guarantees for networks sampled from the CM, but do not have theoretical guarantees for deterministic networks. Nevertheless, we conduct a numerical analysis suggesting that our approach may be exploited to design interventions even in real networks with clustering and other realistic features that are not captured by the CM.}



The paper is organized as follows. Section \ref{Dynamics and Intervention} introduces the LTM and formulates the optimal intervention problem. In Section \ref{Optimal Intervention in large-scale networks}, we reformulate our problem for large-scale networks by the local mean-field approximation. Section \ref{Problem solutions} illustrates how to reduce the obtained optimization problem to a linear problem with finitely many constraints. Section~\ref{Numerical simulations} presents some numerical experiments, and Section \ref{sec:conclusion} summarizes the work.

\textbf{Notation}
    We denote the vectors of all ones and zeros, whose size may be deduced from the context, by $\mathbf{1}$ and $\mathbf{0}$, respectively. For a nonempty finite set $\mc A$, $\mathbb{R}^{\mc A}$ denotes the $|\mc A|$-dimensional space of vectors $x$ whose entries $x_a$ are indexed by elements $a$ of $\mc A$. 
    The transpose of a matrix $A$ in $\mathbb R^{\mc A\times\mc B}$ is denoted by $A'$. Inequalities between vectors are meant to hold true entry-wise, i.e., for $x$ and $y$ in $\mathbb R^{\mc A}$, we write $x\le y$ to indicate that $x_a\le y_a$ for every $a$ in $\mc A$. 

\section{Dynamics and Intervention}\label{Dynamics and Intervention}

We model networks as finite directed multigraphs $\mathcal{N} = (\mathcal{V}, \mathcal{E}, \theta, \lambda)$, where $\mathcal{V}$ is a node set, $\mathcal{E}$ is the set of directed links, and $\theta : \mathcal{E} \to \mathcal{V}$ and $\lambda: \mathcal{E} \to \mathcal{V}$ are two maps associating to each link its tail and head node, respectively. Let $n = |\mathcal{V}|$ be the network order and $A$ in $\mathbb{Z}^{\mathcal{V} \times \mathcal{V}}$ its adjacency matrix, with entries $A_{ij} = |\{e \in \mathcal{E} : \ \theta(e) = i, \lambda(e) = j\}|$. The vectors $\kappa = A\mathbf{1}$ and $\delta = A'\mathbf{1}$ collect the out- and in-degree of the nodes.
We assume that $\mathcal{N}$ contains no self-loops, i.e., $\theta(e) \neq \lambda(e)$ for every $e$ in $\mathcal{E}$.

\subsection{The Linear Threshold Model}
This section introduces the LTM. Within this discrete-time model, each agent $i = 1, \dots, n$ is endowed with a time-varying binary state $x_{i}(t)$ in $\{0, 1\}$ and a threshold $\rho_i$ in $\{0, 1, \dots, \kappa_i\}$, representing the minimum number of state-$1$ neighbors that agent $i$ must observe  in order to adopt state $1$ at the next time step.

Given a network $\mathcal{N} = (\mathcal{V}, \mathcal{E}, \theta, \lambda)$ and a vector of thresholds $\rho$, the LTM is the discrete-time dynamical system on the configuration space $\mathcal{X} = \{0, 1\}^{\mathcal{V}}$ defined by
\begin{equation}\label{LTD}
    x(t + 1) = \Phi_\rho(x(t))\,, \qquad \forall t \geq 0\,,
\end{equation}
    
where $\Phi_\rho : \mathcal{X} \to \mathcal{X}$ is the map defined by
\begin{equation}\label{eq:Phi}
    (\Phi_\rho(x))_i = 
    \begin{cases}
         1 & \text{if} \ \sum\nolimits_{j \in \mathcal{V}} A_{ij} x_j \geq \rho_i\\
         0 & \text{if} \ \sum\nolimits_{j \in \mathcal{V}} A_{ij} x_j < \rho_i
    \end{cases}
 \end{equation}
\begin{remark}\label{remark:monotone} For every network $\mc N$ and threshold vector $\rho$, the map $\Phi_\rho$ is non-decreasing, i.e., it preserves the partial ordering on $\mc X$. Hence, given two ordered initial conditions $x(0) \le \bar x(0)$, we have that $x(t) \le \bar x(t)$ for every time $t\ge 0$.
\end{remark}
\begin{remark}
\tcb{In the seminal paper by \cite{Kempe.ea:03}, the term \emph{linear threshold model} refers to \eqref{LTD} with random thresholds distributed uniformly at random. We remark that in this paper the thresholds are arbitrary.}
\end{remark}

\subsection{Intervention and optimization problem}
This section introduces the notion of intervention and formalizes the optimization problem that the planner aims to solve.
We assume that the planner can modify the agent thresholds.
Given a network $\mathcal{N} = (\mathcal{V}, \mathcal{E}, \theta, \lambda)$, threshold vector $\rho$, and  vector $h$ in $\mathbb{Z}^{\mathcal{V}}_{+}$, to be referred to as an \emph{intervention}, the \emph{LTM with intervention} is
\begin{equation}\label{LTD-intervention}
    x(t + 1) = \Phi_{\rho - h}(x(t))\,, \qquad \forall t \geq 0\,.
\end{equation}
In plain words, setting intervention $h$ decreases the threshold of each player $i$ by $h_i$ units.  Since the thresholds must remain non-negative, an intervention is said feasible if
\begin{equation}\label{eq:feasible}
h \le \rho\,.
\end{equation}
An intervention $h$ is associated with a separable cost $\sum_{i \in \mc V} \cost_i(h_i)$, where $\cost_i: \mathbb Z_+ \to \mathbb R_+$ is the cost function of node $i$, assumed to be non-decreasing and such that $\cost_i(0) = 0$ for every $i$ in $\mc V$.

The goal of the planner is to design a minimal cost intervention such that, for a given $\epsilon$ in $\tcb{(0, 1]}$,
\begin{equation}\label{fraction of state - 1 agents}
    \frac{1}{n} \sum_{i \in \mathcal{V}}x_i(t) \geq 1 - \epsilon \ , \qquad \forall t \geq n,
\end{equation}
for every $x(0)$ in $\mathcal{X}$, i.e., the fraction of agents in state $1$ after $t \ge n$ iterations is at least $1-\epsilon$, for every initial condition $x(0)$ in $\mathcal{X}$. By \eqref{LTD-intervention} and by Remark \ref{remark:monotone}, requiring that \eqref{fraction of state - 1 agents} holds true for every $x(0)$ in $\mc X$ is equivalent to
\begin{equation}\label{fraction of state - 1}
    \frac{1}{n} \sum_{i \in \mathcal{V}}((\Phi_{\rho-h})^{n}(\mathbf{0}))_{i}\geq 1 - \epsilon\,.
\end{equation}
Hence, the optimization problem for the planner is
\begin{equation}\label{optimal control}
    \min\limits_{h \in \mathbb{Z}_+^{\mathcal{V}}} \sum_{i \in \mathcal{V}} \cost_i(h_i) \ \ \text{s.t.} \ \eqref{eq:feasible},\eqref{fraction of state - 1}\,.
\end{equation}

\begin{remark}\label{remark:seeding}
\tcb{Problem \eqref{optimal control} is a variant of the LCIP where the links have integer weights. The original LCIP is known to be NP-complete, even in very simplified settings, see \cite{gunnecc2020least}.}
Observe that, when the cost functions are such that $\cost_{i}(h_{i}) = c_i$ for a positive constant $c_i$ for every $h_{i} > 0$ and node $i$ in $\mc V$, then it is optimal to set either the null intervention $h_{i} = 0$ or $h_{i} = \rho_{i}$ (so that necessarily $x_{i}(t) = 1$ for every $t>0$). In this case, the optimal intervention problem reduces to a \tcb{variant of the WTSS problem}. If, in particular, $c_i = 1$ for every node $i$ in $\mc V$, the problem reduces to finding a minimal cardinality set of agents such that, if they are forced to adopt state $1$, then for all $t \geq n$ the fraction of agents in state $1$ is at least $1 - \epsilon$, which is \tcb{a variant of the TSS problem}. 
\end{remark}

\subsection{Agent types}
In large-scale systems, detailed knowledge of the network topology and of the other features of the agents may be unavailable, while some aggregate statistical information might be accessible.
For this reason, it is convenient to consider the distribution of agents with certain features by assigning to every agent $i$ in $\mc V$ a type $\omega_i$ that uniquely determines its in-degree $\delta_i$, out-degree $\kappa_i$, threshold $\rho_i$ and cost function $\cost_i \ : \mathbb{Z}_{+} \rightarrow \mathbb{R}_+$. Let $\Omega$ denote a countable set of agent types and let $d, k, r$ in $\mathbb{R}^{\Omega}$ and $c \ : \mathbb{Z}_{+} \rightarrow \mathbb{R}^{\Omega}_+$ be such that
\begin{equation*}
d_{\omega_i} = \delta_i, \quad \! k_{\omega_i} = \kappa_i, \quad \! r_{\omega_i} = \rho_i, \quad \! c_{\omega_{i}}(\cdot) = \cost_i(\cdot), \quad \! \forall i \in \mathcal{V}\,.
\end{equation*}
We then define the vector $\statistics^{(0)}$ in $[0,1]^{\Omega}$ by
\begin{equation}\label{statistics}
    \statistics_{w}^{(0)} = |\{i \in \mathcal{V} \ : \ \omega_i = w\}| / n\,, \qquad \forall w \in \Omega\,,
\end{equation}
denoting the empirical distribution of types, referred to as the \emph{statistics} of the problem before the intervention.

Analogously, we can associate to every intervention $h$ a vector-valued function $\xi: \mathbb Z_+ \to [0,1]^{\Omega}$, with $\control_w(\eta)$ denoting the fraction of agents whose type before the intervention is $w$ and whose threshold is decreased by $\eta$ units by the intervention. Consistently with such interpretation and with \eqref{eq:feasible}, $\control$ must satisfy for every type $w$ in $\Omega$
\begin{equation}\label{eq:feasible_xi}
    \sum_{\eta = 0}^{r_w} \control_w(\eta) = \statistics_{w}^{(0)},\qquad \xi_w (\eta) = 0\,, \ \forall \eta > r_w\,.
\end{equation}
Moreover, since the number of nodes of type $w$ whose threshold is decreased by $\eta$ units due to the intervention has to be an integer, $\control$ must also satisfy
\begin{equation}\label{consistent - control}
	n \control_{w}(\eta) \in \mathbb{Z}_{+}, \qquad \forall \eta = 0, 1, \cdots, r_{w}\,,\ \forall w \in \Omega\,.
\end{equation}
We call a pair $(n,\xi)$ feasible if conditions \eqref{eq:feasible_xi}-\eqref{consistent - control} hold.

Now, let
$
\Omega_w = \{\tilde w \in \Omega: d_{\tilde w}=d_w, k_{\tilde w}=k_w , c_{\tilde w}=c_w , r_{\tilde w} > r_w\}
$
denote the set of types $\tilde w$ that differ from type $w$ only in the threshold and such that $r_{\tilde w} > r_w$.
By decreasing the agent thresholds, an intervention generates new network statistics $\statistics{(\control)}$ defined by
\begin{equation}\label{statistics-control}
    \statistics_{w}{(\control)} = \statistics_{w}^{(0)} + \!\!\! \sum_{\tilde w \in \Omega_w} \!\!\! \control_{\tilde w}(r_{\tilde w} - r_{w}) - \sum\limits_{\eta = 1}^{r_{w}} \control_{w}(\eta)\,,
\end{equation}
for every $w$ in $\Omega$. In particular, the second term accounts for the agents of types $\tilde w$ with $r_{\tilde w} > r_w$ whose threshold is decreased by $r_{\tilde w} - r_{w}$ units, and the last term is the fraction of agents of type $w$ whose threshold is decreased, so that the type of those agents varies.

\begin{remark}
The null statistical intervention, denoted $\control^{(0)}$, has entries $$\control_{w}^{(0)}(\eta) = \begin{cases} \statistics_{w}^{(0)}\,, \quad & \text{if } \eta = 0,\\ 
0 & \text{otherwise}\,.\end{cases}$$
With such intervention, $\statistics(\control^{(0)}) = \statistics^{(0)}\,.$
\end{remark}

\section{Optimal Intervention in large-scale networks}\label{Optimal Intervention in large-scale networks}
\tcb{Problem \eqref{optimal control} presents two main issues. First, it requires full network knowledge. Second, even if the network is known, the complexity scales badly with the network size.} To deal with these issues, we adopt the following approach. Using the local mean-field approach proposed in \cite{Rossi.ea:17}, we approximate the evolution of the fraction of agents in state $1$ under the LTM by a 1-dimensional recursive equation with output, and then reformulate the optimization in this framework. This approach, previously adopted by \cite{messina2024optimal} to address the optimal seeding problem in super-modular games (\tcb{which is strictly related to the WTSS problem}), has two main advantages. First, it allows to find solutions whose complexity scales well with the network size. Second, it relies on statistical knowledge instead of full knowledge of the network structure, of the thresholds, and of the cost functions of all agents.

\subsection{Local mean-field approximation of linear threshold model}
In order to sample uniformly from the set of problems with given statistics, we introduce the following notions. For an arbitrary probability distribution $p$ in $[0, 1]^{\Omega}$ and a vector $f$ in $\mathbb{R}_{+}^{\Omega}$, we let
\begin{equation*}
    \left\langle \statistics, f \right\rangle = \sum_{w \in \Omega} \statistics_{w} f_w\,.
\end{equation*}

\subsubsection{Well posedness:}
Given a statistics $\statistics(\control)$ and a finite set of nodes $\mathcal{V}$ of cardinality $n = |\mathcal{V}|$, a vector of types $\omega$ in $\Omega^{\mathcal{V}}$ with statistics $p(\xi)$ exists if and only if
\begin{equation}\label{consistent-integer}
     n \statistics_{w}(\control) \in \mathbb{Z}_+ \qquad \forall w \in \Omega.
\end{equation}  
Furthermore, for the existence of a network $\mathcal{N}$ with node set $\mathcal{V}$, in- and out-degrees $\delta_i = d_{\omega_i}$ and $\kappa_i = k_{\omega_i}$ for all $i$ in $\mathcal{V}$, and without selfloops, it is necessary and sufficient that
\begin{equation}\label{consistent-1moment}
    \langle \statistics(\control), d \rangle = \langle \statistics(\control), k \rangle\,,
\end{equation}
\begin{equation}\label{consistent-noselfloop}
    d_{w} + k_w \leq n \langle \statistics(\control), d \rangle,\qquad \forall w \in \Omega \ : \ \statistics_{w}(\control) > 0.
\end{equation}
Equation \eqref{consistent-1moment} guarantees that the total out-degree equals the total in-degree, while \eqref{consistent-noselfloop} ensures the existence of a set of links compatible with the degree distribution and that does not contain any self-loops.

We say that a pair $(n, \statistics(\control))$ is feasible if conditions \eqref{consistent-integer}-\eqref{consistent-noselfloop} are met. In particular, this is guaranteed if both $(n,\xi)$ and $(n,p^{(0)})$ are feasible pairs. 

We now describe our procedure to sample uniformly from the set of problems with $n$ nodes and statistics $p(\xi)$ for a feasible pair $(n,p(\xi))$. The procedure is divided in two steps. We first create $n$ nodes with a type distribution $\omega$ with statistics $p(\xi)$, and then choose a wiring uniformly at random to create the links of the network.
\begin{definition}\label{configuration model}
    Let $\mathcal{V} = \{1,\ldots, n\}$ and let the type profile $\omega$ in $\Omega^{\mathcal{V}}$ have empirical distribution $\statistics(\control)$. Let $l=n\langle p(\xi),d\rangle$ and $\mc E=\{1,\ldots,l\}$. Define $\theta:\mc E\to\mc V$ by letting  $\theta(e)$ in $\mc V$ be the unique value such that 
    $\sum_{i=1}^{\theta(e)-1}\kappa_{i}<e\le \sum_{i=1}^{\theta(e)}\kappa_{i}$, for all $e$ in $\mc E$. 
    Let $\lambda:\mc E\to\mc V$,  $\lambda(e)=\theta(\pi(e))$ for  $e$ in $\mc E$, where $\pi$ is sampled uniformly from the set of permutations of $\mc E$ such that 
    \begin{equation}\label{permutation-noloops}\theta(e)\ne\theta(\pi(e))\,,\qquad\forall e\in\mc E\,.\end{equation}
    Then, the \emph{(resampled) configuration model} $\mc C_{n,p(\xi)}$ is the triple $(\mc N,\rho,\gamma)$ of the random network $\mc N=(\mc V,\mc E,\theta,\lambda)$, the threshold vector $\rho$, and the vector of cost functions $\gamma$  with entries $\rho_i=r_{w_i}$ and $\gamma_i(\cdot)=c_{w_i}(\cdot)$, respectively.  
\end{definition}



Following the approach of \cite{messina2024optimal}, in the rest of this section we shall reformulate the intervention problem \eqref{optimal control}, expressed in terms of interventions $h$, in terms of the statistics $\xi$.
The next result is instrumental towards this goal. To better formalize it, we define
\begin{equation*}
    \varphi_{kr}(z) = \sum\limits_{u = r}^k \binom{k}{u} z^u (1-z)^{k-u}, \qquad \forall  0 \leq r \leq k\,,
\end{equation*}
and two maps $\psi_{\statistics(\control)},\phi_{\statistics(\control)}:[0,1] \to [0,1]$ by
\begin{equation}
    \psi_{\statistics(\control)}(z)  = \sum_{w\in\Omega} \statistics_{w}(\control)\varphi_{k_wr_w}(z)\,,
\end{equation}
\begin{equation}\label{eq:phi}
    \phi_{\statistics(\control)}(z)  =\frac{1}{\langle \statistics(\control), d \rangle} \sum_{w\in\Omega} \statistics_{w}(\control)d_w\varphi_{k_wr_w}(z)\,.
\end{equation}
These two maps are polynomials with
non-negative coefficients that depend on the network statistics $\statistics(\control)$, hence, in particular they are monotonically non-decreasing.

We can now formalize the main result of this section, which is at the basis of the problem reformulation. 
\begin{proposition}\label{prop:constraint}
    Let $\statistics(\control)$ be a probability distribution on $\Omega$ with finite moments
    $\langle \statistics(\control), d \rangle$, $\langle \statistics(\control), k \rangle$, $\langle \statistics(\control), d^2 \rangle$, and $\langle \statistics(\control), k^2 \rangle$.
    Let $(n, \statistics^{(n)}(\control))$ be a sequence of compatible pairs such that
    \begin{align*}
        \statistics^{(n)}(\control) &\stackrel{n\to+\infty}{\longrightarrow}\statistics(\control)\,,\\
        \langle\statistics^{(n)}(\control), d\rangle &\stackrel{n\to+\infty}{\longrightarrow}\langle \statistics(\control), d \rangle,\\
        \langle\statistics^{(n)}(\control), d^2 \rangle  &\stackrel{n\to+\infty}{\longrightarrow} \langle \statistics^{(n)}(\control), d^2\rangle,\\
        \langle\statistics^{(n)}(\control), k^2\rangle  &\stackrel{n\to+\infty}{\longrightarrow}\left\langle \statistics(\control), k^2\right\rangle.
    \end{align*}
	Let $\epsilon$ in $\tcb{(0,1]}$. If
    \begin{equation}\label{phi-ineq}
        \tcb{\phi_{\statistics(\control)}(z) > z}, \quad \forall z \in [0,\psi_{\statistics(\control)}^{-1}(1 - \epsilon)]\,,
    \end{equation}
    then \eqref{fraction of state - 1} holds true on a fraction of problems drawn from the (resampled) configuration model ensemble $\mathcal{C}_{n, \statistics^{(n)}(\control)}$ that converges to $1$ as $n$ grows large.
\end{proposition}

\begin{pf}
A similar result was already proved by \cite{messina2024optimal} \tcb{for the WTSS problem}. For completeness of presentation, we report its proof in Appendix \ref{app:1}.
\hfill $\blacksquare$
\end{pf}

Proposition \ref{prop:constraint} states that, given an intervention $\xi$ with associated statistics $p(\xi)$ satisfying \eqref{phi-ineq}, the dynamics will converge almost surely  to a configuration with at least a fraction $1-\epsilon$ of agents in state $1$, as far as the network is sufficiently large. This in particular implies that almost all interventions $h$ with associated statistics $\xi$, which have the same cost by construction, will drive the system towards a desirable configuration. This result is the key to reformulate the intervention problem in terms of $\xi$.

The idea behind the proof of Proposition \ref{prop:constraint} is that the fraction of state-$1$ adopters in the LTM can be approximated by the output $y(t)$ of 1-dimensional recursion of the form
\begin{equation}\label{recursion}
    y(t+1) = \psi_{\statistics(\control)}(z(t))\,, \qquad
	z(t+1) = \phi_{\statistics(\control)}(z(t))\,.
\end{equation}
As \cite{Rossi.ea:17} showed, $z(t)$ and $y(t)$ may be interpreted, respectively, as the fraction of links that point to state-$1$ agents and the fraction of state-$1$ agents, whose dynamics, in expectation with respect to the sampling from the configuration model ensemble, depend on the intervention $\xi$ by the functions $\phi_{p(\xi)}$ and $\psi_{p(\xi)}$. In particular, \eqref{phi-ineq} and the right-most equation in \eqref{recursion} ensure that the fraction of links pointing towards state-$1$ agents will be at least $\psi_{\statistics(\control)}^{-1}(1 - \epsilon)$ after a finite number of iterations. Thus, the left-most equation in \eqref{recursion} ensures that the fraction of nodes in state $1$ will be at least $1-\epsilon$, achieving the goal of the planner. Note also that, according to this interpretation, the fact that $\phi_{p(\xi)}$ and $\psi_{p(\xi)}$ are non-decreasing is consistent with the fact that $\Phi_\rho$ defined in \eqref{eq:Phi} is monotone, as observed in Remark \ref{remark:monotone}. An example of dynamics of $z(t)$ with initial condition $z(0) = 0$ is illustrated in Figure \ref{recursion figure}. For this example, the fraction of links pointing to state-$1$ agents (and therefore also the fraction of state-$1$ agents) converges to $1$. \tcb{For more details and intuition on the derivation of \eqref{recursion} we refer to \cite{Rossi.ea:17}}.
\begin{figure}
    \centering
    \includegraphics[scale=0.2]{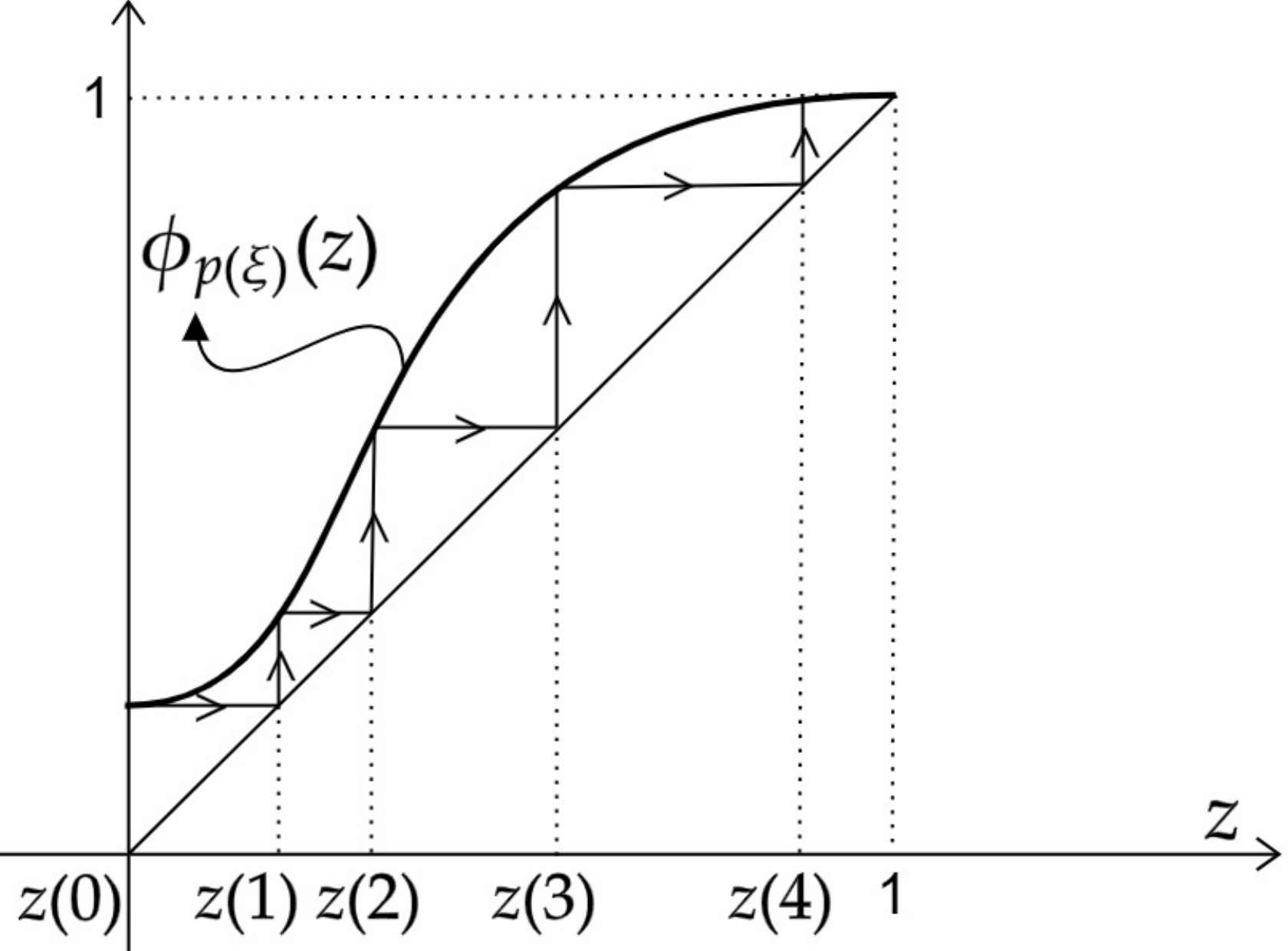}
	\caption{The evolution of fraction of links pointing to state-$1$ agents in the local mean-field local approximation according to \eqref{recursion} for a statistics $p(\xi)$. \label{recursion figure}}
\end{figure}

\subsection{Problem reformulation}
In this section, we reformulate the optimization problem \eqref{optimal control} in the statistical framework. In this framework, the goal of the planner is to find the minimal cost intervention $\xi$ such that \eqref{phi-ineq} holds true, so that Proposition \ref{prop:constraint} guarantees that on almost all large-scale networks with statistics $p(\xi)$ at least a fraction $1-\epsilon$ of nodes will converge to state $1$.

Since $c_{w}(\eta)$ is the cost for decreasing by $\eta$ units the threshold of an agent of type $w$, it follows by definition of $\control$ that the number of nodes of type $w$ whose threshold is decreased by $\eta$ units is $n \control_{w}(\eta)$. Therefore, it is natural to associate to intervention $\xi$ a cost 
$$C(\control) = \sum\limits_{w \in \Omega}\sum\limits_{\eta = 0}^{r_{w}} \control_{w}(\eta) c_w(\eta) \ge 0\,,$$
with $C(\xi^{(0)}) = 0$.
We can now reformulate our optimization problem into
\begin{equation}\label{optimization problem}
	\underset{\control: \mathbb Z_+ \to [0,1]^{\Omega} }{\inf} \ C(\control) \ \
	\text{s.t.} \ \eqref{eq:feasible_xi}, \eqref{phi-ineq}\,.
\end{equation}
The next section analyzes the properties of \eqref{optimization problem} and provides a numerical method to solve it.

\section{Problem solutions}\label{Problem solutions}
This section focuses on problem \eqref{optimization problem}. The next result is instrumental for its analysis.

\begin{proposition}\label{controlled recursion}
	For every intervention $\xi: \mathbb Z_+ \to [0,1]^{\mc V}$ that satisfies \eqref{eq:feasible_xi},
	\begin{equation}\label{eq:prop2}
		\phi_{\statistics(\control)}(z) = \phi_{p^{(0)}}(z) + \sum_{w \in \Omega} \sum\limits_{\eta = 1}^{r_w} a_{w}(\eta, z) \control_{w}(\eta),     
	\end{equation} 
	where 
	\begin{equation}\label{eq:a}
	a_w(\eta, z) = \frac{d_w  \left( \varphi_{k_w(r_w - \eta)}(z) - \varphi_{k_wr_w}(z) \right)}{\langle p^{(0)}, d \rangle} \geq 0\,.
	\end{equation}
\end{proposition}

\begin{pf}
Note that, since the intervention modifies only the agent thresholds and not the degrees, $\langle p(\xi),d \rangle = \langle p^{(0)},d \rangle$. Hence, it follows from \eqref{eq:phi} that
	\begin{equation}\label{eq:phi_p_eps}
		\phi_{\statistics(\control)}(z)  = \frac{1}{\langle p^{(0)}, d \rangle} \sum_{w\in\Omega} \statistics_{w}(\control)d_w\varphi_{k_wr_w}(z)
	\end{equation}
	Plugging \eqref{statistics-control} into this expression, we get that the right-hand side of \eqref{eq:phi_p_eps} is equal to the sum of three terms:
		$$\text{term 1} = \frac{1}{\langle p^{(0)}, d \rangle} \sum_{w\in\Omega} \statistics_{w}^{(0)}d_w\varphi_{k_wr_w}(z) = \phi_{{\statistics}^{(0)}}(z)\,,$$
		\begin{align*}
		\text{term 2}& =\frac{1}{\langle p^{(0)}, d \rangle} \sum_{w\in\Omega} \sum_{w' \in \Omega_w} \control_{w'}(r_{w'} - r_{w})d_w\varphi_{k_w, r_w}(z) \\
			& = \frac{1}{\langle p^{(0)}, d \rangle} \sum_{w\in\Omega} \sum\limits_{\eta = 1}^{r_{w}} \control_{w}(\eta) d_{w}\varphi_{k_w, r_w - \eta}(z)\,,
		\end{align*}
		$$\text{term 3} = - \frac{1}{\langle p^{(0)}, d \rangle} \sum_{w\in\Omega} \sum\limits_{\eta = 1}^{r_{w}} \control_{w}(\eta) d_{w}\varphi_{k_w, r_w}(z)\,.$$
	Summing the contributions, we obtain \eqref{eq:prop2}. The fact that $a_w(\eta,z) \ge 0$ is consequence of the definition of $\varphi_{kr}(z)$. \hfill $\blacksquare$
\end{pf}

Constraint \eqref{phi-ineq} is linear in $\control$  since $\phi_{p(\xi)}$ is linear in $\xi$ by Proposition \ref{controlled recursion}.  However, problem \eqref{optimization problem} remains challenging beacause the domain on which constraint \eqref{phi-ineq} must be satisfied depends on the intervention variable $\control$ through  $\psi_{p(\xi)}$. Moreover, constraint \eqref{phi-ineq} must be satisfied for a continuum of $z$, leading to an infinite-dimensional optimization problem.
In this section we propose a numerical method to address both these challenges and find close to optimal solutions of \eqref{optimization problem} for small $\epsilon > 0$.

\begin{remark}
\tcb{While the original LCIP \eqref{optimal control} is a non-linear program whose size scales with the number of agents, its local mean-field reformulation is a linear program whose size scales with the number of types. We remark that the local mean-field reformulation of the WTSS proposed in \cite{messina2024optimal} is also a linear program, but the two problems have different sizes, since in \eqref{optimization problem} the planner chooses for every type $w$ the fraction of agents whose is decreased by $\eta$ units, with $\eta = 0,\cdots,r_w$, while in the WTSS $\eta \in \{0,r_w\}$. The advantage of the LCIP is that the objective of steering the network towards a desirable configuration can be achieved with considerably smaller budget, as the example in Section \ref{Numerical simulations} illustrates.}
\end{remark}
    
To tackle the first challenge, we define an alternative optimization problem that will prove to be a relaxation of the former one. Consider the constraint
\begin{equation}\label{eq:new_constr}
\tcb{\phi_{\statistics(\control)}(z) > z}\,, \quad \forall z \in [0, 1 - \alpha_\epsilon]\,,
\end{equation}
where $d_{\min} = \min_{i \in \mathcal{V}} \delta_i$ and 
\begin{equation}\label{eq:alpha}
\alpha_\epsilon = \epsilon d_{\min} / \langle \statistics^{(0)}, d \rangle\,.
\end{equation}
Now, consider the optimization problem
\begin{equation}\label{optimization problem relaxed}
	\underset{\control: \mathbb Z_+ \to [0,1]^{\Omega} }{\inf} \ C(\control) \ \
\text{s.t.} \ \eqref{eq:feasible_xi},  \eqref{eq:new_constr}\,.
\end{equation}
This problem is easier to solve than \eqref{optimization problem} since $\alpha_\epsilon$ does not depend on the intervention variable $\xi$.
The next result illustrates the relation between the two problems.
\begin{proposition}\label{convergence relaxed problem}
    Let $\mathcal{P}_{\epsilon}$ denote the set of feasible $\control$ for problem \eqref{optimization problem} and $\mathcal{P}_{\alpha_{\epsilon}}$ denote the set of feasible $\control$ for problem \eqref{optimization problem relaxed}. Then:
	\begin{enumerate}
        \item[i)] $\mathcal{P}_{\alpha_{\epsilon}} \subseteq \mathcal{P}_{\epsilon}$.
        \item[ii)] $\lim_{\epsilon \to 0}  \mathcal{P}_{\epsilon} \setminus \mathcal{P}_{\alpha_{\epsilon}} = \emptyset$.
	\end{enumerate} 
\end{proposition}

\begin{pf}
A similar result was first established in \cite{messina2024optimal}. For completeness of presentation, we report it in Appendix \ref{app:2}.\hfill $\blacksquare$
\end{pf}

While the two optimization problems \eqref{optimization problem} and \eqref{optimization problem relaxed} have the same objective function, Proposition \ref{convergence relaxed problem} establishes that the set of feasible interventions for the former problem contains the set of feasible interventions for the latter one, hence the solution of \eqref{optimization problem relaxed} is feasible for \eqref{optimization problem}. Moreover, for small values of $\epsilon$, the two sets of feasible interventions converge to each other. This suggests that the solutions of problem \eqref{optimization problem relaxed}, beyond being feasible for \eqref{optimization problem}, are also close to optimal as $\epsilon$ vanishes. This is key for our main result Theorem \ref{thm:uN}.

We now focus on problem \eqref{optimization problem relaxed}, and propose a discretization to obtain a finite-dimensional problem. In particular, we discretize the interval $[0, 1 - \alpha_{\epsilon}]$ in $N + 1$ equally spaced points $$z_i = (1 - \alpha_{\epsilon})i/N\,, \quad i = 0, 1, \cdots, N\,,$$ and require that
\begin{equation}\label{eq:constraint_discrete}
\phi_{\statistics(\control)}(z_i) - z_i \geq \Delta\,, \quad \forall i = 0,1,\cdots,N\,,
\end{equation}
for a properly chosen $\Delta>0$. 

\begin{lemma}\label{lemma:der}
For every $\control$ that satisfies \eqref{eq:feasible_xi},
    \begin{equation*}
        \left|\frac{\de}{\de z} (\phi_{\statistics(\control)}(z) - z) \right| \le \frac{d_{\max} 2^{k_{\max}+1} k_{\max}}{\left\langle \statistics^{(0)}, d\right\rangle} + 1\,, \quad \forall z \in [0,1]\,.
    \end{equation*}
    where $k_{\max} = \max_{i \in \mathcal{V}} \kappa_i$ and $d_{\max} = \max_{i \in \mathcal{V}} \delta_i$.
\end{lemma}

\begin{pf}
Using the fact that $\langle p(\epsilon),d\rangle = \langle p^{(0)},d\rangle$ and the definition of $\phi_{\statistics(\control)}$, it follows that
   \begin{equation}\label{eq:proof}
        \!\!\! \begin{aligned}
            &\!\left|\frac{\de}{\de z} (\phi_{\statistics(\control)}(z)) \right| = \left| \sum_{w}\frac{d_w}{\left\langle \statistics^{(0)}, d\right\rangle} \statistics_w(\control) \frac{\de}{\de z}\varphi_{k_wr_w}(z) \right|.
        \end{aligned} 
    \end{equation}
    Moreover,
    \begin{align*}
        \!\left|\frac{\de}{\de z} \varphi_{k_wr_w}(z)\right| \! &= \! \sum\limits_{u = r_w}^{k_w} \!\! \binom{k_w}{u} z^{u-1}(1 - z)^{k_w - u - 1}\left|-k_w z + u \right| \\
        & \leq \! \sum\limits_{u = r_w}^{k_w} \!\! \binom{k_w}{u} \!\left|-k_w z + u \right| \leq 2 \sum\limits_{u = r_w}^{k_w} \!\! \binom{k_w}{u} k_{\max} \\
        & \leq 2 k_{\max}  \sum_{u = 0}^{k_w} \binom{k_w}{u} = 2 k_{\max} 2^{k_{\max}},
    \end{align*}

    where the two inequalities follow from the fact that $z$ belongs to $[0,1]$ and the last equality from $2^{k_{\max}} = (1+1)^{k_{\max}} = \sum_{u=0}^{k_{\max}} \binom{k_w}{u}$. Plugging this expression into \eqref{eq:proof}, we get
    \begin{align*}
        \left|\frac{\de}{\de z} (\phi_{\statistics(\control)}(z) - z) \right| &\le \frac{2^{k_{\max} + 1} k_{\max}}{\langle \statistics^{(0)}, d\rangle} \sum_{w} d_w \statistics_w(\control) + 1 \\
        &\le \frac{2^{k_{\max}+1} k_{\max} d_{\max}}{\langle \statistics^{(0)}, d\rangle} \sum_w \statistics_w(\control) + 1\\
        & = \frac{2^{k_{\max}+1} k_{\max} d_{\max}}{\langle \statistics^{(0)}, d\rangle} + 1\,,
    \end{align*}
proving the statement. \hfill $\blacksquare$
\end{pf}

Given $N$ in $\mathbb N$, we now define the discretized problem
\begin{equation}\label{optimization problem discretized}
	\underset{\control: \mathbb Z_+ \to [0,1]^{\Omega} }{\inf} \ C(\control) \ \
\text{s.t.} \ \eqref{eq:feasible_xi}, \eqref{eq:constraint_discrete}\,.
\end{equation}
The next result establishes the relation between the discretized problem \eqref{optimization problem discretized} and the original problem \eqref{optimization problem}.  

\begin{thm}\label{thm:uN}
    Let $\epsilon>0$, $N$ in $\mathbb N$ and $\Delta \ge \Delta_N$, with
    \begin{equation}\label{eq:deltaN}
    \Delta_N := \frac{1 - \alpha_\epsilon}{2 N} \left( \frac{d_{\max} 2^{k_{\max}+1} k_{\max}}{\left\langle \statistics^{(0)}, d\right\rangle} + 1 \right)\,.    
    \end{equation} 
    Let $\control^{*}$ and $\xi^N$ be the solutions of \eqref{optimization problem} and \eqref{optimization problem discretized}, respectively.
    Then:
    \begin{enumerate}
        \item[i)] $\control^N$ is feasible for \eqref{optimization problem};
        \item[ii)] if $\Delta = \Delta_N$, $\lim_{\epsilon \to 0} \lim_{N \to +\infty} C(\control^N) = C(\control^*)$.
	\end{enumerate}
\end{thm}
    
\begin{pf}
    i) Since $\control^{N}$ is solution of \eqref{optimization problem discretized}, $\xi^N$ needs to satisfy \eqref{eq:constraint_discrete} with $\Delta_N$ defined in \eqref{eq:deltaN}, that is, $$\phi_{\statistics(\control^{N})}(z_i) - z_i \geq \Delta\,, \quad \forall i = 0,\cdots, N\,.$$ 
    Since $z_{i+1} - z_i = (1 - \alpha_{\epsilon})/N$, standard properties of the first derivative imply that, for every $z$ in $[0, 1 - \alpha_{\epsilon}]$,
    $$
        \tcb{\phi_{\statistics\left(\control^{N}\right)}(z)-z > \Delta - \frac{(1-\alpha_\epsilon)}{2N} \left|\frac{\de}{\de z} \left(\phi_{\statistics\left(\control^{N}\right)}(z)-z\right)\right|}.
    $$
    Since $\Delta \ge \Delta_N$, by \eqref{eq:deltaN} and by Lemma \ref{lemma:der}, we get that $$\tcb{\phi_{\statistics\left(\control^{N}\right)}(z) - z > 0}\,,\quad \forall z \in [0,1-\alpha_\epsilon]\,.$$ Hence, $\control^{N}$ satisfies \eqref{eq:new_constr} and is therefore feasible for problem \eqref{optimization problem relaxed}. Proposition \ref{convergence relaxed problem}(i) then implies that $\xi^N$ is also feasible for problem \eqref{optimization problem}.
		
    ii) Let $\control^{\alpha_\epsilon}$ be the solution of problem \eqref{optimization problem relaxed}. Consider the case when there exists a type $q$ in $\Omega$ with positive empirical frequency $p_q(\xi^{\alpha_\epsilon})>0$ and positive threshold $r_q>0$. Now, let 
    \begin{equation}\label{eq:def:beta}
    \beta_N = \frac{\Delta_N}{\min_{\zeta \in [0,\alpha_\epsilon]} a_{q}(r_{q}, \zeta)}\,,
    \end{equation}
    well defined because $a_w(\eta,z)>0$ for every $z$ in $[0,\alpha_\epsilon]$, $\eta>0$ and type $w$ in $\Omega$ (cf. \eqref{eq:a}). 
    Note now from \eqref{eq:deltaN} and \eqref{eq:def:beta} that 
    \begin{equation}\label{eq:beta}
    \Delta_N \stackrel{N \to +\infty}{\longrightarrow} 0\,, \qquad \beta_N \stackrel{N \to +\infty}{\longrightarrow} 0\,.
    \end{equation}
    Let $q'$ in $\Omega$ be the type differing from $q$ only in the threshold value $r_{q'} = 0$.
    Hence, there always exists a sufficiently large $N$ such that
    \begin{equation}\label{eq:gamma_tilde}
    	p_w(\hat{\control}) = 
    	\begin{cases}
    		p_w(\xi^{\alpha_\epsilon}) - \beta_N, &\text{if} \ w = q,\\
    		p_w(\xi^{\alpha_\epsilon}) + \beta_N, &\text{if} \ w = q',\\
    		p_w(\xi^{\alpha_\epsilon}) & \text{otherwise}\,,
    	\end{cases}
    \end{equation}
    satisfies $\statistics(\hat{\control}) \in[0,1]^{\Omega}$ and is therefore a statistics.
    Moreover, Proposition \ref{controlled recursion} implies that, for every $z$ in $[0, 1 - \alpha_{\epsilon}]$,
    \begin{equation}\label{eq:phi_gamma_hat}
    	\begin{aligned}
    		\phi_{\statistics(\hat{\control})}(z) & = \phi_{\statistics(\xi^{\alpha_\epsilon})}(z) + a_{q}(r_{q}, z) \cdot \beta_N \\ 
    		& \geq \phi_{\statistics(\xi^{\alpha_\epsilon})}(z) + \Delta_N \geq z + \Delta_N,
    	\end{aligned}
    \end{equation}
    where the last inequality follows from the fact that $\xi^{\alpha_\epsilon}$ is feasible for problem \eqref{optimization problem relaxed} (in fact, it is its solution) and therefore satisfies \eqref{eq:new_constr}. Eq. \eqref{eq:phi_gamma_hat} implies that $\hat{\control}$ satisfies \eqref{eq:constraint_discrete} and is therefore feasible for problem \eqref{optimization problem discretized}, whose solution is $\control^N$. Moreover, $\control^N$ is feasible for problem \eqref{optimization problem relaxed}, whose solution is $\xi^{\alpha_\epsilon}$. Therefore,
    \begin{equation}\label{eq:bounds}
    	C(\hat{\control}) \geq C(\control^N) \geq C(\control^{\alpha_\epsilon})\,.
    \end{equation} 
    Now, note from \eqref{eq:gamma_tilde} that $$C(\hat{\control}) = C(\xi^{\alpha_\epsilon}) + c_q(w_q) \beta_N \stackrel{N \to +\infty}{\longrightarrow} C(\control^\alpha)\,.$$
    This, together with \eqref{eq:CN}, proves that
    \begin{equation}\label{eq:CN}
    	C(\control^N) \stackrel{N \to +\infty}{\longrightarrow} C(\control^{\alpha_\epsilon})\,.
    \end{equation}  
    Finally, Proposition \ref{convergence relaxed problem}(ii) and continuity of $C(\xi)$ together prove that
    $$ 
    C(\xi^{\alpha_\epsilon}) \stackrel{\epsilon \to 0}{\longrightarrow} C(\xi^*)\,.
    $$
    Hence, for this case the proof is concluded by combining this equation with \eqref{eq:CN}. 
    
    We now consider the case when $p(\xi^{\alpha_\epsilon})$ is such that $p_w(\xi^{\alpha_\epsilon}) = 0$ for every type $w$ in $\Omega$ with positive threshold $r_w > 0$. Observe from \eqref{eq:alpha} that $\alpha_{\epsilon} > 0$ because $\epsilon > 0$. Since $\xi^{\epsilon_\alpha}$ is optimal for problem \eqref{optimization problem relaxed} and $\alpha_\epsilon>0$, it follows that $\xi^{\alpha_\epsilon} = \xi^{(0)}$. Therefore, it follows from \eqref{eq:phi} that $\phi_{p^{(0)}}(z) = 1$ for every $z$ in $[0,\alpha_\epsilon]$. This, together with \eqref{eq:beta} and with the fact that $C(\xi^{(0)}) = 0$, implies that, as $N$ grows large, $\xi^N = \xi^* = \xi^{(0)}$. The proof is now concluded.
%
%
     \hfill $\blacksquare$
\end{pf}

Theorem \ref{thm:uN} states that the solutions of \eqref{optimization problem discretized} are feasible and close to optimal for problem \eqref{optimization problem} as long as $N$ grows large and $\epsilon$ vanishes. Remarkably, \eqref{optimization problem discretized} is a linear program with finitely many constraints whose size scales with the number of agent types, much easier to solve than the original statistical problem \eqref{optimization problem}, and, even more importantly, of problem \eqref{optimal control}.

\section{Numerical simulations}\label{Numerical simulations}
\begin{figure}
	\centering
	\includegraphics[scale=0.20]{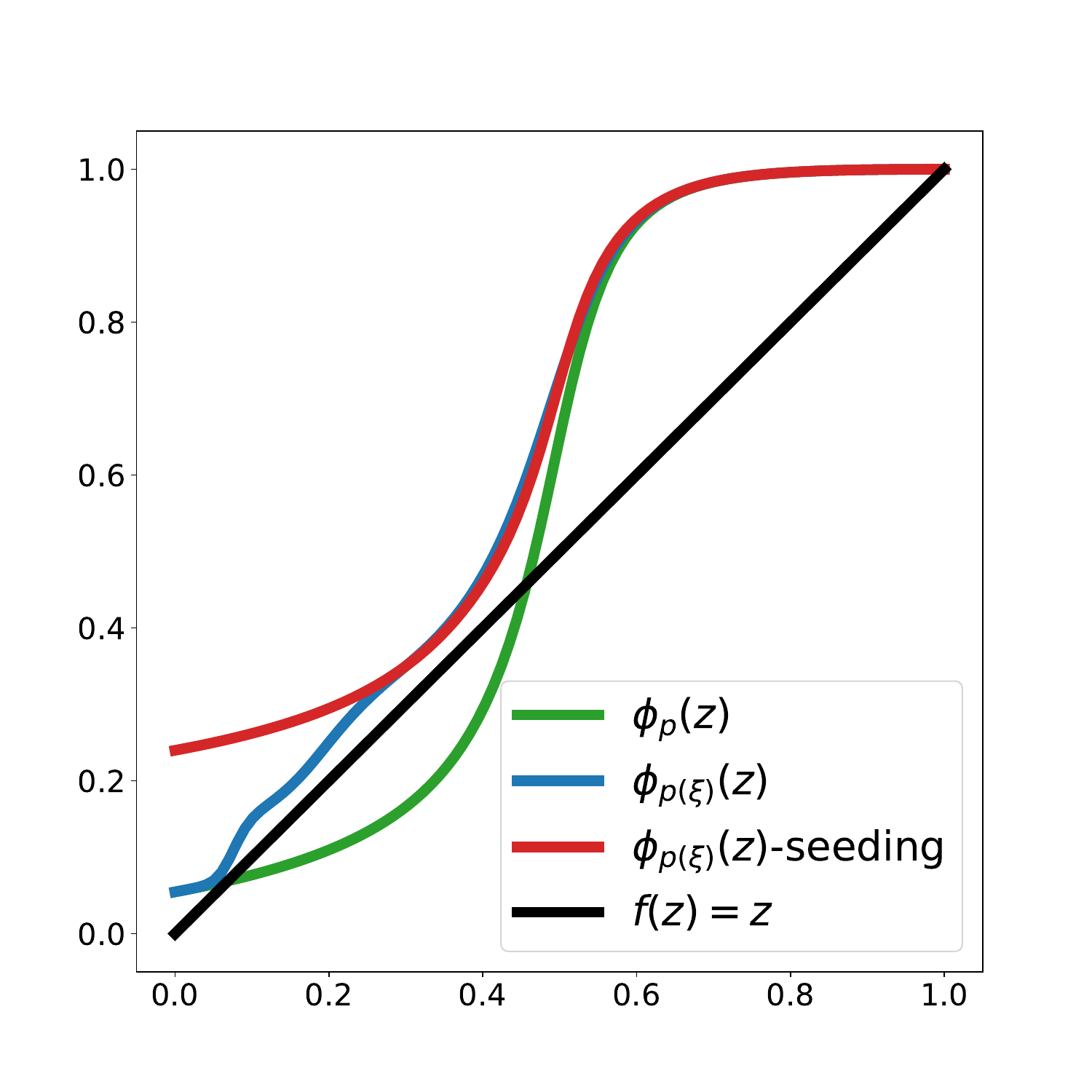}
	\caption{Evolution function for the fraction of links pointing to state-$1$ agents with initial statistics $p^{(0)}$, with  optimal intervention statistics $p(\xi^N)$, and with optimal seeding $p(\xi^{N,s})$, for the Epinions network.} \label{fig:interventions}
\end{figure}

\tcb{In this section, we present numerical experiments on the real-world networks Epinions\footnote{Retrieved from https://networkrepository.com} and Power Grid Network.\footnote{Retrieved from https://websites.umich.edu/$\sim$ mejn/netdata/} }

\tcb{The network Epinions contains $26588$ nodes and $100120$ undirected links. 
For every type $w$ in $\Omega$, we set threshold $r_w = \left\lfloor k_w/2 \right\rfloor$ and cost $c_w(\eta) = \eta$. This defines a triple $(\mc N,\rho,\gamma)$. We extract the statistics $p^{(0)}$ and derive the optimal statistical intervention $\xi^N$ by solving problem \eqref{optimization problem discretized} with $\epsilon = 0.1$, $N = 100$ and $\Delta = 0.05$ and the optimal statistical intervention $\control^{N,s}$ for the related WTSS problem studied in \cite{messina2024optimal}. Such problem is obtained by setting $c_w(0) = 0$ and $c_w(\eta) = r_w$ for $\eta > 0$, so that for the planner it is optimal either not to intervene on an agent or to set her threshold to $0$.
We then implement a random intervention that complies with the optimal statistical one, and compare the fraction of state-$1$ agents observed on the network under such interventions with the fraction of state-$1$ agents predicted by recursion \eqref{recursion}.}

\tcb{Fig. \ref{fig:interventions} illustrates the recursion function driving the evolution of the fraction of links pointing to state-$1$ agents in the local mean-field approximation with the initial statistics $p^{(0)}$, with statistics $p(\xi^N)$ corresponding to the optimal intervention, and with the statistics $p(\xi^{N,s})$ corresponding to the optimal seeding. Intuitively speaking, the seeding is suboptimal whenever the curve $\phi_{\statistics^{(0)}}(z)$ is above $z$ for small values of $z$ and below $z$ for large values of $z$. While in such a case the intervention $\xi^N$ allows to target the curve in precise regions of the domain, seeding interventions are forced to increase the entire curve, especially for small values of $z$, producing suboptimal intervention costs.}
	
\tcb{To confirm the validity of the statistical approach, we then implement random interventions on the network that comply with the statistical interventions $\xi^N$ and $\xi^{N,s}$ and simulate the TM with initial condition $x(0) = \mathbf 0$ on the Epinions network. Figure \ref{fig:SDynamic-recursion} illustrates the evolution of the fraction of state-$1$ agents on the networks versus the one predicted by recursion \eqref{recursion} based only on the statistics. Note that the seeding intervention, due to its higher cost, is able to drive the agents towards state $1$ faster than the optimal intervention. However, also the optimal intervention drives the system towards the desired configurations in a finite number of steps, as predicted by \eqref{recursion}. We remark that our interventions are designed without relying on the exact network knowledge, suggesting that designing interventions based on the problem statistics is a valid approach even for real networks that are not generated from a configuration model ensemble. Moreover, our interventions apply to an arbitrary threshold distribution, in contrast with those provided by \cite{Kempe.ea:03}, which assumes random and uniformly distributed thresholds.}
\begin{figure}
	\centering
	\includegraphics[scale=0.20]{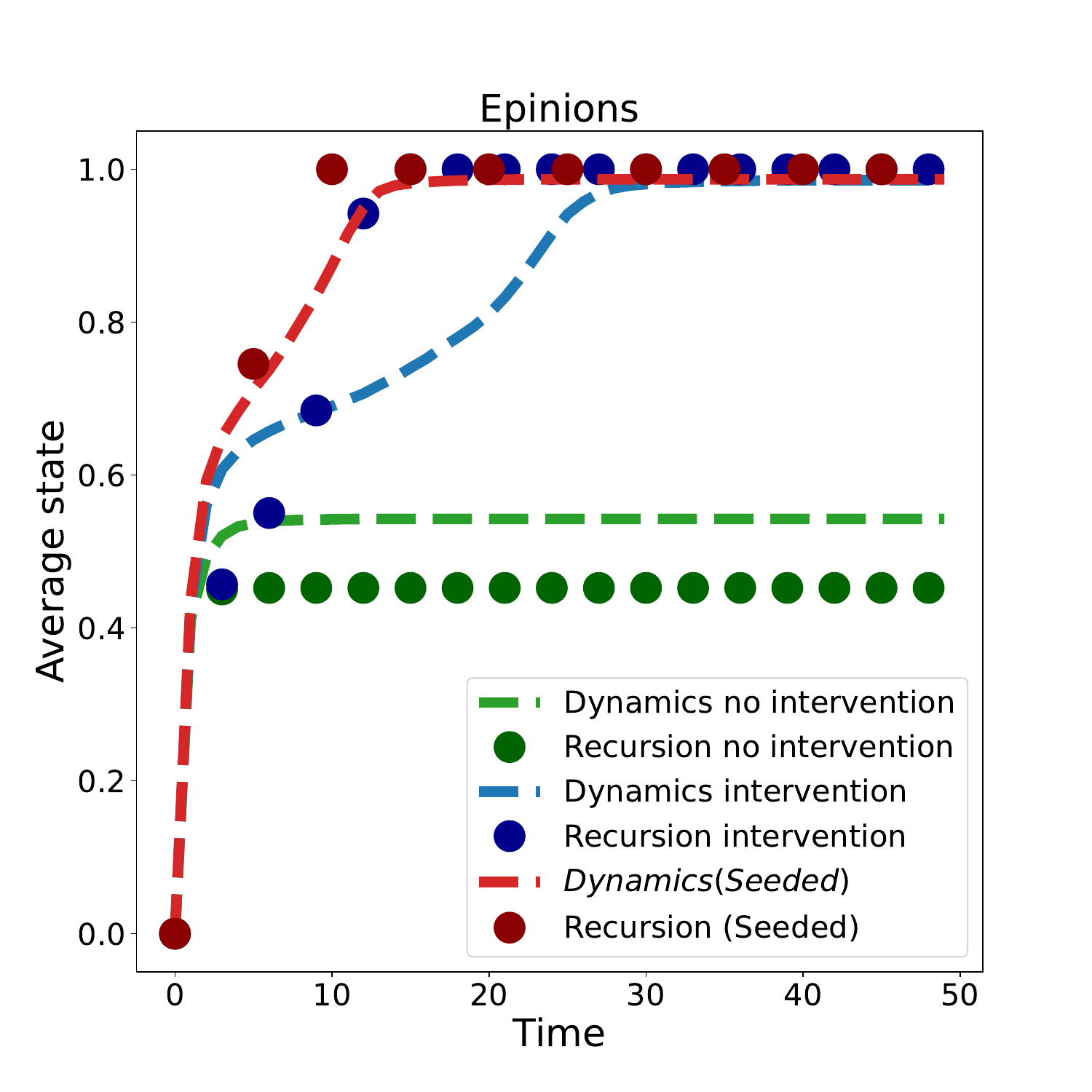}
	\caption{The evolution of the fraction of agents in state $1$ on the \tcb{Epinions network} versus the one predicted by the recursion \eqref{recursion} for the example of Section \ref{Numerical simulations}. \label{fig:SDynamic-recursion}}
\end{figure}

\tcb{We then conduct a numerical experiment on the Power Grid Network\footnote{Retrieved from https://websites.umich.edu/$\sim$ mejn/netdata/}. This is an undirected, unweighted network representing the topology of the Western States Power Grid of the United States, with $4941$ nodes and $6594$ undirected links. The LCIP problem on this network was studied in \cite{Cordasco:2015} with thresholds $\rho_i$ chosen uniformly at random in $\{1, \dots, \kappa_i\}$, with cost functions $\gamma_i(\eta) = \eta$ and with $\epsilon = 0.3$. We extract the statistics $p^{(0)}$ and derive $\xi^N$ by solving problem \eqref{optimization problem discretized} with $N = 100$ and $\Delta = 0.05$.
By averaging over $10$ instances of random thresholds, we observe that our interventions achieve the desired fraction of agents in state $1$ and outperform those proposed by \cite{Cordasco:2015} reducing the intervention cost by about $25\%$, despite $N$ being quite small. We remark that our solutions do not rely on the exact network knowledge, as instead done by \cite{Cordasco:2015}.}

\section{Conclusion}\label{sec:conclusion}
We have studied the design of optimal solutions of the LCIP problem in the LTM in large-scale networks, where a planner can modify agent thresholds to guarantee that a predetermined fraction of agents adopt state-$1$ in finite time. We built on a local mean-field approximation of the dynamics to formulate the corresponding intervention problem as a linear program with an infinite number of constraints, relying on a statistical rather than exact network knowledge. We then propose a method to transform the resulting infinite programming into a finite linear program with a dimension that does not scale with the network size. The proposed methodology is validated numerically. We expect that variants of our techniques can be applied to different problems and network dynamics.


\section*{DECLARATION OF GENERATIVE AI AND AI-ASSISTED TECHNOLOGIES IN THE WRITING PROCESS}
During the preparation of this work, AI-based tools were used solely for language proofreading and did not contribute to the scientific content of the manuscript.

\section*{FUNDING}
This work was partially supported by the MUR-funded Research Project PRIN 2022 ``Extracting
Essential Information and Dynamics from Complex Networks'' (Grant  Agreement number 2022MBC2EZ).

\bibliography{ref}             
                                                   







\appendix
\section{Proof of Proposition \ref{prop:constraint}}\label{app:1}
Let $\nu = \left\langle \statistics(\control), dk \right\rangle / \left\langle \statistics(\control), d \right\rangle - 1$. Finiteness of the first and second moments of $\statistics(\control)$ implies that $0 \leq \nu < +\infty$.
An argument analogous to \cite[Theorem 7.2]{VanDerHofstad:16} implies that the probability that a uniform random permutation $\pi$ satisfies \eqref{permutation-noloops} converges to $e^{-\nu/2}$ as $n$ grows large.
    Let $Y(t) = \frac{1}{n} \sum_i x_i(t)$ and $Z(t) = \frac{1}{l} \sum_i(\delta_i x_i(t))$ denote the fraction of state-$1$ agents and the fraction of links pointing to state-$1$ agents in the LTM \eqref{LTD-intervention}. Define $z(t)$ and $y(t)$ recursively by putting $y(0) = z(0) = 0$ and 
    \begin{equation}\label{recursion threshold model}
        y(t+1) = \psi_{\statistics(\control)}(z(t)) \qquad
		z(t+1) = \phi_{\statistics(\control)}(z(t)),
    \end{equation}
for $t\geq 0$. 
    Then, \cite[Theorem 1]{Rossi.ea:17} implies that $Z(t)$ and $Y(t)$ are arbitrarily close to $z(t)$ and $y(t)$, respectively, on all but an asymptotically vanishing fraction of networks $\mathcal{N} = (\mathcal{V}, \mathcal{E}, \theta, \lambda)$, where $\lambda(e) = \theta(\pi(e))$ for a uniform random permutation $\pi$ of $\mathcal{E}$. Since $\pi$ satisfies \eqref{permutation-noloops} with probability bounded away from $0$ asymptotically, this implies that $Z(t)$ and $Y(t)$ are arbitrarily close to $z(t)$ and $y(t)$, respectively, on the configuration model $\mathcal{C}_{n,\statistics^{(n)}(\control)}$ with probability approaching $1$ as $n$ grows large.
    Observe that \eqref{phi-ineq} and the second equation in \eqref{recursion threshold model} imply that there exists a finite time $\tau$ such that \tcb{$z(t) > \psi_{\statistics(\control)}^{-1}(1-\epsilon)$} for all $t \geq \tau$.
    The first equation in \eqref{recursion threshold model} then implies that \tcb{$y(t) > 1-\epsilon$} for $t > \tau$. This implies that, with probability converging to $1$ as $n$ grows large $\frac{1}{n} \sum_i x_i(t) = Y(t) \tcb{> 1 - \epsilon}$ for $t > \tau$, i.e., \eqref{fraction of state - 1 agents} holds true, thus completing the proof.
\section{Proof of Proposition \ref{convergence relaxed problem}}\label{app:2}
    Following the steps of the proof of Proposition \ref{prop:constraint}, constraint \eqref{eq:new_constr} implies that with high probability the dynamics converges to a configuration with at most a fraction $\alpha_{\epsilon}$ of links pointing to agents with state $0$ so that the fraction of agents with state $0$ is upper-bounded by $\alpha_{\epsilon} \langle p, d \rangle / d_{\min} = \epsilon$. This in turn proves (c.f. proof of Proposition \ref{prop:constraint} and \cite{Rossi.ea:17}) that \tcb{$\psi_{\statistics(\control)}(z) > 1 - \epsilon$}, so that \tcb{$z > \psi_{\statistics(\control)}^{(-1)}(1 - \epsilon)$} and \eqref{phi-ineq} is satisfied. 
    To prove ii), notice that $\psi_{\statistics(\control)}^{-1}(1) = 1$ for all $\control$. By continuity of $\psi_{\statistics(\control)}$ and by definition of $\alpha_{\epsilon}$, $$\psi_{\statistics(\control)}^{-1}(1 - \epsilon) \stackrel{\epsilon \to 0^{+}}{\longrightarrow} 1\,, \quad 1 - \alpha_\epsilon \stackrel{\epsilon \to 0^{+}}{\longrightarrow} 1\,,$$ so that the sets of feasible interventions for \eqref{optimization problem} and \eqref{optimization problem relaxed}, respectively,  get arbitrarily close as $\epsilon$ vanishes. 
\end{document}